\documentclass[envcountsect,envcountsame,11pt]{llncs}
\usepackage[a4paper,hmargin=1.0in,vmargin=1.0in]{geometry}

\usepackage{framed,xspace}
\usepackage{bbm,qip}
\usepackage{color,graphicx}
\usepackage{ifpdf}

\usepackage{tikz}

\newcommand{\A}{{\sf A}}
\newcommand{\dA}{{\sf A}'}
\newcommand{\B}{{\sf B}}
\newcommand{\dB}{{\sf B}'}

\newcommand{\hA}{\hat{\sf A}}
\newcommand{\dhA}{\hat{\sf A}'}
\newcommand{\hB}{\hat{\sf B}}
\newcommand{\dhB}{\hat{\sf B}'}
\newcommand{\F}{\mathcal{F}}
\newcommand{\ID}{{\sf \textsl{ID}}\xspace}

\newcommand{\OT}[1]{{\sf \textsl{OT}}\ensuremath{^{#1}}\xspace}
\newcommand{\eqq}{\stackrel{\raisebox{-1ex}{\tiny?}}{=}}
\newcommand{\approxq}{\stackrel{\text{\it\tiny q}}{\approx}}
\newcommand{\approxs}{\stackrel{\text{\it\tiny s}}{\approx}}
\newcommand{\set}[1]{\{#1\}}
\newcommand{\Set}[2]{\{ #1 : #2\}}
\newcommand{\eps}{\varepsilon}
\newcommand{\x}{\times}

\newcommand{\commit}{{\tt Commit}}
\newcommand{\pk}{{\tt pk}}
\newcommand{\sk}{{\tt sk}}
\newcommand{\pkH}{{\tt pkH}}
\newcommand{\pkB}{{\tt pkB}}
\newcommand{\GH}{{\cal G}_{\tt H}}
\newcommand{\GB}{{\cal G}_{\tt B}}
\newcommand{\compile}{{\cal C}^{\alpha}}

\newcommand{\Hoo}{H_{\infty}}
\newcommand{\Hmin}{H_{\infty}}
\newcommand{\Hmax}{H_0}

\newcommand{\MC}[3]{#1 \leftrightarrow #2 \leftrightarrow #3}

\newcommand{\ignore}[1]{}

\newcommand{\X}{\mathcal{X}}
\newcommand{\crs}{\omega}

\newcommand{\dBob}{\mathfrak{B}}

\newcommand{\dBobBQSM}{\dBob_{\text{\sc bqsm}}}
\newcommand{\dBobPoly}{\dBob_{\mathrm{poly}}}

\newcommand{\poly}{\mathit{poly}}

\newcommand{\qp}{\pi}
\newcommand{\cp}{\Sigma}
\newcommand{\nord}{{}}
\newcommand*{\assign}{\ensuremath{\kern.5ex\raisebox{.1ex}{\mbox{\rm:}}\kern -.3em =}}

\newcommand{\QID}{\mathtt{QID}}

\newcommand{\QOT}{{\tt 1\text{-}2 \, QOT^\ell}}

\newcommand{\test}{t\hspace{-0.15ex}e\hspace{-0.1ex}st}
\newcommand{\Test}{T\hspace{-0.3ex}e\hspace{-0.15ex}st}
\newcommand{\err}{e\hspace{-0.1ex}r\hspace{-0.15ex}r}

\newcommand{\AppendixOnOff}[1]{#1}
\ifpdf

\else

\fi

\begin{document}
\allowdisplaybreaks
\pagestyle{plain}

\title{Improving the Security of Quantum Protocols via Commit-and-Open}

\author{Ivan Damg{\aa}rd\inst{1} \and Serge Fehr\inst{2} \and Carolin Lunemann\inst{1} \and Louis
  Salvail\inst{3} \and Christian Schaffner\inst{2}}

\institute{
  DAIMI, Aarhus University, Denmark\\
  \email{\{ivan|carolin\}@cs.au.dk} \and
  Centrum Wiskunde \& Informatica (CWI) Amsterdam, The Netherlands\\
  \email{\{s.fehr|c.schaffner\}@cwi.nl} \and
  Universit\'e de Montr\'eal (DIRO), QC, Canada\\
  \email{salvail@iro.umontreal.ca} }

\maketitle

\begin{abstract}
We consider two-party quantum protocols starting with a transmission of
some random BB84 qubits followed by classical messages. We show a general
``compiler'' improving the security of such protocols: if the original protocol is
secure against an ``almost honest'' adversary, then the compiled protocol
is secure against an arbitrary computationally bounded (quantum) adversary.
The compilation preserves the number of qubits sent and the number of rounds up to a
constant factor.
The compiler also preserves security in the bounded-quantum-storage model (BQSM), 
so if the original protocol was BQSM-secure, the compiled protocol can only be broken 
by an adversary who
has large quantum memory {\em and} large computing power. This is in contrast to known BQSM-secure protocols, where security breaks down completely if the adversary has larger quantum memory than expected. We show how our technique can be applied to quantum identification and oblivious transfer protocols. \\[1ex]
%{\bf Keywords.} Quantum cryptography, bounded-quantum-storage and
%computational security, oblivious transfer and identification. 
\end{abstract}

\section{Introduction}
We consider two-party quantum protocols for mutually
distrusting players Alice and Bob.  Such protocols typically start by
Alice sending $n$ random BB84 qubits to Bob who is supposed to measure
them. Then some classical exchange of messages follows. Several
protocols following this pattern have been proposed, implementing
Oblivious Transfer (OT), 
Commitment, and Password-Based Identification
\cite{BBCS91,DFSS08,DFRSS07,DFSS07}.

In more details, the first step of the protocol consists of Alice
choosing random binary strings $x= x_1,...,x_n$ and $\theta =
\theta_1,...,\theta_n$. She then prepares $n$ particles where $x_i$ is
encoded in the state of the $i$'th particle using basis
$\theta_i$. Bob chooses a basis string $\hat{\theta} =
\hat{\theta}_1,..,\hat{\theta} _n$ and measures the $i$'th particle in
basis $\hat{\theta}_i$. If Bob plays honestly, he learns $x_i$
whenever $\hat{\theta}_i= \theta_i$ and else gets a random
independent result.

Protocols of the form we consider here are typically unconditionally
secure against cheating by Alice, but can (in their basic form) be
broken easily by Bob, if he does not measure the qubits
immediately. This is because the protocol typically asks Alice to
reveal $\theta$ at a later stage, and Bob can then measure the qubits
with $\hat{\theta} = \theta$ and learn more information than he was
supposed to.

In this paper, we show a general ``compiler'' that can be used to
improve security against such an attack. We assume that the original
protocol implements some two-party functionality $\F$ with statistical
security against Bob if he is {\em benign}, meaning that he treats the
qubits ``almost honestly'', a notion we make more precise below. Then
we show that the compiled protocol also implements $\F$, but now with
security against {\em any} computationally bounded (quantum) Bob (note
that we cannot in general obtain unconditional security against both
Alice and Bob, not even using quantum communication~\cite{Lo97}).  The
compiled protocol preserves unconditional security against Alice and
has the same number of transmitted qubits and rounds as the original
one up to a constant factor.
 
By benign behavior of Bob, we mean that after having received the
qubits, two conditions are satisfied: First, Bob's quantum storage is
essentially of size zero (note that it would be exactly zero if he had
measured the qubits).  Second, there exists a basis string
$\hat{\theta}$ such that the uncertainty about $x$ is essentially as
it would be if Bob had really measured in bases $\hat{\theta}$,
namely 1 bit for every position where $\hat{\theta} $ differs
from~$\theta$.

Thus, with our compiler, one can build a protocol for any two-party
functionality by designing a protocol that only has to be secure if
Bob is benign. We note that proofs for known protocols typically go
through under this assumption.  For instance, our compiler can easily
be applied to the quantum identification protocols of \cite{DFSS07}
and the OT protocol of \cite{BBCS91}.

The compiler is based on a computational assumption; namely we assume
the existence of a classical commitment scheme with some special
properties, similar to the commitment schemes used in \cite{DFS04} but
with an additional extraction property, secure against a quantum
adversary.  A good candidate is the cryptosystem of
Regev~\cite{Regev05}.  For efficiency, we use a common reference
string which allows us to use Regev's scheme in a simple way and,
since it is relatively efficient, we get a protocol that is
potentially practical.  It is possible to generate the reference
string from scratch, but this requires a more complicated non-constant
round protocol~\cite{DL09arxiv}.

The reader may ask whether it is really interesting to improve the
security of quantum protocols for classical tasks such as
identification or OT using a computational assumption. Perhaps it
would be a more practical approach to use the same assumption to build
{\em classical} protocols for the same tasks, secure against quantum
attacks?  To answer this, it is important to point out that our
compiler also preserves security in the bounded-quantum-storage model
(BQSM) \cite{DFSS05}, and this feature allows us to get security
properties that classical protocols cannot achieve. In the BQSM, one
assumes that Bob can only keep in his quantum memory
a limited number of qubits received
from Alice.  With current state of the art, it
is much easier to transmit and measure qubits than it is to store them
for a non-negligible time, suggesting that the BQSM and the
subsequently proposed noisy-quantum-storage model~\cite{WST08} are
reasonable. On the other hand, if the assumption fails and the
adversary can perfectly store all qubits sent, the known protocols can
be easily broken.  In contrast, by applying our compiler, one obtains
new protocols where the adversary must have large quantum storage {\em
  and} large computing power to break the protocol.%
\footnote{For the case of identification\cite{DFSS07}, the compiled
  protocol is not only secure against adversaries trying to
  impersonate Alice or Bob, but can also be made secure against
  man-in-the-middle attacks, where again the adversary must have large
  quantum storage and large computing power to break the protocol.}
\footnote{One may try to achieve the same security by combining one of
  the previous BQSM secure protocols with a computationally secure
  classical protocol, but it is not clear that this technique will
  work for all functionalities, and it would require independent key
  material for the two instances. For the case of password-based
  identification it would require users to have two passwords.}

The basic technique we use to construct the compiler was already
suggested in connection with the first quantum OT protocol
from~\cite{BBCS91}: we try to force Bob to measure by asking him to
commit (using a classical scheme) to all his basis choices and
measurement results, and open some of them later.  While classical
intuition suggests that the commitments should force Bob to measure
(almost) all the qubits, it has proved very tricky to show that the
approach really works against a quantum adversary.  In fact, it was
previously very unclear what exactly the commit-and-open approach
forces Bob to do.  Although some partial results for OT have been
shown~\cite{Yao95,CDMS04}, the original OT protocol from \cite{BBCS91}
has never been proved secure for a concrete unconditionally hiding
commitment scheme -- which is needed to maintain unconditional
security against Alice.  In this paper, we develop new quantum
information-theoretic tools (that may be of independent interest) to
characterize what commit-and-open achieves in general, namely it
forces Bob to be benign. This property allows us to apply the compiler
to any two-party functionality and in particular to show that the OT
from \cite{BBCS91} is indeed secure when using an appropriate
commitment scheme.

\section{Preliminaries}
We assume the reader to be familiar with the
basic notation and concepts of quantum information processing~\cite{NC00}.
In this paper, the computational or $+\,$-basis is defined by the pair
$\{ \ket{0}, \ket{1} \}$ (also written as $\{ \ket{0}_+, \ket{1}_+
\}$). The pair $\{ \ket{0}_\times, \ket{1}_\times \}$ denotes the
diagonal or $\times$-basis, where \smash{$\ket{0}_\times = (\ket{0} +
\ket{1}) / \sqrt{2}$} and \smash{$\ket{1}_\times = (\ket{0} - \ket{1}) /
\sqrt{2}$}. We write $\ket{x}_\theta = \ket{x_1}_{\theta_1} \otimes
\cdots \otimes \ket{x_n}_{\theta_n}$ for the $n$-qubit state where
string $x = (x_1,\ldots,x_n) \in \{0,1\}^n$ is encoded in bases
$\theta = (\theta_1,\ldots,\theta_n) \in \{+,\times\}^n$. For $S
\subseteq \{ 1, \ldots, n \}$ of size $s$, we denote by $\bar{S}
\assign \set{1,\ldots,n} \backslash S$ the complement of $S$ and define $x|_S \in \{ 0,1
\}^s$ and $\theta|_S \in \{ +,\times \}^s$ to be the restrictions
$(x_i)_{i \in S}$ and $(\theta_i)_{i \in S}$, respectively. For two
strings $x,y \in \set{0,1}^n$, we define the \emph{Hamming distance}
between $x$ and $y$ as $d_H(x,y) \assign
\left|\Set{i}{x_i \neq y_i}\right|$.

We use upper case letters for the random variables in the
proofs that describe the respective values in the protocol.
Given a bipartite quantum state $\rho_{XE}$, we say that $X$ is {\em
  classical} if $\rho_{XE}$ is of the form $\rho_{XE} = \sum_{x \in
  \X} P_X(x) \proj{x} \otimes \rho_E^x$ for a probability distribution
$P_X$ over a finite set $\X$, i.e.~the state of the quantum register
$E$ depends on the classical random variable $X$ in the sense that
$E$ is in state $\rho_E^x$ exactly if $X = x$. This naturally extends to states with two or more classical registers. 

For a state $\rho_{XE}$ as above, $X$ is {\em independent} of register $E$ if $\rho_{XE} = \rho_X \otimes \rho_E$, where $\rho_X = \sum_x P_X(x) \proj{x}$ and $\rho_E = \sum_x P_X(x) \rho_E^x$. 
We also need to express that a random variable $X$ is independent of a
quantum state $E$ {\em when given a random variable $Y$}. Independence
means that when given $Y$, the state $E$ gives no additional
information on $X$. Formally, adopting the notion introduced
in~\cite{DFSS07}, we require that $\rho_{X Y E}$ equals
$\rho_{X\leftrightarrow Y \leftrightarrow E}$, where the latter is
defined as
$$
\rho_{X\leftrightarrow Y \leftrightarrow E} \assign \sum_{x,y}P_{X
  Y}(x,y)\proj{x} \otimes \proj{y} \otimes \rho_{E}^y \, ,
$$ 
where $\rho_{E}^y \assign \sum_x P_{X|Y}(x|y) \rho_E^{x,y}$.
In other words, $\rho_{X Y E} = \rho_{X\leftrightarrow Y
  \leftrightarrow E}$ precisely if $\rho_E^{x,y} = \rho_E^{y}$ for all
$x$ and $y$.

Full (conditional) independence is often too strong a requirement, and
it usually suffices to be ``close'' to such a situation. Closeness of
two states $\rho$ and $\sigma$ is measured in terms of their trace
distance $\delta(\rho,\sigma) = \frac{1}{2} \tr(|\rho-\sigma|)$, where
for any operator $A$, $|A|$ is defined as \smash{$|A| \assign \sqrt{A A^\dag}$}.

A quantum algorithm consists of a family $\{ C_n\}_{n \in \naturals}$
of quantum circuits and is said to run in \emph{polynomial time}, if
the number of gates of $C_n$ is polynomial in $n$. Two families of
quantum states $\set{\rho_n}_{n \in \naturals}$ and $\set{\sigma_n}_{n
  \in \naturals}$ are called \emph{quantum-computationally
  indistinguishable}, denoted $\rho \approxq \sigma$, if any
polynomial-time quantum algorithm has negligible advantage (in $n$) of
distinguishing $\rho_n$ from $\sigma_n$. Analogously, we call them
\emph{statistically indistinguishable}, $\rho \approxs \sigma$, if
their trace distance $\delta(\rho_n,\sigma_n)$ is negligible in~$n$.

\begin{definition}[Min-Entropy]
The {\em min-entropy} of a random variable $X$ with probability 
distribution~$P_X$ is defined as $H_{\infty}(X) \assign -\log\bigl( \max_x P_X(x) \bigr)$.
\end{definition}

\begin{definition}[Max-Entropy]
The {\em max-entropy} of a density matrix $\rho$ is defined as $H_0(\rho) \assign \log\bigl( \rank{\rho}\bigr)$.
\end{definition}
We will make use of the following properties of a pure state that can
be written as a ``small superposition'' of basis vectors.

\begin{lemma}\label{lemma:SmallSuperpos}
Let $\ket{\varphi_{AE}} \in {\cal H}_A \otimes {\cal H}_E$ be of the form 
$\ket{\varphi_{AE}} = \sum_{i \in J} \alpha_i \ket{i}\ket{\varphi_E^i}$, where 
$\set{\ket{i}}_{i \in I}$ is a basis of ${\cal H}_A$ and $J \subseteq I$. Then, 
the following holds.
\begin{enumerate}
\item\label{it:boundHoo} Let $\tilde{\rho}_{AE} = \sum_{i \in J} |\alpha_i|^2 \proj{i}\otimes\proj{\varphi_E^i}$, and let $W$ and $\tilde{W}$ be the outcome of
 measuring $A$ of $\ket{\varphi_{AE}}$ respectively of $\tilde{\rho}_{AE}$ 
in some basis $\set{\ket{w}}_{w \in \cal W}$. Then,%
\footnote{Using Renner's definition for conditional min-entropy~\cite{Renner05}, one can actually show that $H_{\infty}(W|E) \geq H_{\infty}(\tilde{W}|E) - \log|J|$. }
$$
H_{\infty}(W) \geq H_{\infty}(\tilde{W}) - \log|J| \, .
$$
\item\label{it:boundHo} The reduced density matrix $\rho_E = \tr_A(\proj{\varphi_{A E}})$ has max-entropy
$$
H_0(\rho_E) \leq \log |J| \, .
$$
\end{enumerate}
\end{lemma}

 \begin{proof}
 For \ref{it:boundHoo}., we may understand $\tilde{\rho}_{AE}$ as being in state $\ket{i}\ket{\varphi_E^i}$ with probability $|\alpha_i|^2$, so that we easily see that
 \begin{align*}
 P_{\tilde{W}}(w) &= \sum_{i \in J} |\alpha_i|^2 |\braket{w}{i}|^2 
 = \sum_{i \in J} |\alpha_i|^2 |\braket{w}{i}|^2 \cdot \sum_{i \in J} 1^2 \cdot \frac{1}{|J|} \\
 &\geq \bigg|\sum_{i \in J} \alpha_i \braket{w}{i}\bigg|^2 \cdot \frac{1}{|J|} 
 = \bigg|\bra{w}\sum_{i \in J} \alpha_i \ket{i}\bigg|^2 \cdot \frac{1}{|J|} 
 = P_{W}(w) \cdot \frac{1}{|J|} \, ,
 \end{align*}
 where the inequality is Cauchy-Schwartz. This proves \ref{it:boundHoo}.
 
 For \ref{it:boundHo}., note that $\rho_E = \tr_A( \proj{\varphi_{A E}} ) = \sum_{i
   \in J} |\alpha_i|^2 \proj{\varphi_E^i}$. 
   The claim follows immediately from the sub-additivity of the rank: 
 $$
 \rank{\rho_E} \leq \sum_{i \in J} \rank{|\alpha_i|^2 \proj{\varphi_E^i}} \leq  \sum_{i \in J} 1 = |J| \, ,
 $$
 where we use that the $\proj{\varphi_E^i}$'s have rank at most 1.
 \qed
 \end{proof}

\section{Definition of Security} \label{sec:security}

In order to define security of our two-party protocols, we follow the
framework put forward by Fehr and Schaffner in~\cite{FS09}. We are
interested in quantum protocols that implement \emph{classical
  functionalities} such as oblivious transfer. Such primitives are
often used as building blocks in more complicated classical
(multi-party) protocols which implement advanced tasks. Therefore, it
is natural to restrict our focus on quantum protocols that run in a
classical environment and have classical in- and outputs.  A two-party
quantum protocol $\Pi = (\A_m,\B_m)$ consists of an infinite family of
interactive quantum circuits for players Alice and Bob indexed by the
security parameter $m$ (in our case, $m$ will also be the number of
qubits transmitted). To ease notation, we often leave the dependence
on $m$ implicit.  A classical non-reactive two-party {\em ideal 
functionality} $\F$ is given by a conditional probability
distribution $P_{\F(U,V)|UV}$, inducing a pair of random variables
$(X,Y) = \F(U,V)$ for every joint distribution of $U$ and $V$. The
definition of correctness of a protocol is straightforward.
\begin{definition}[Correctness] \label{def:correctness}
 A protocol $\Pi = (\A,\B)$ \emph{correctly implements} an ideal
 classical functionality $\F$, if for every distribution of the input
 values $U$ and $V$, the resulting common output
satisfies
\[ (U,V,(X,Y)) \approxs (U,V, \F(U,V)) \, .
\]
\end{definition}

Let us denote by $out_{\hA,\hB}^\F$ the joint output%
\footnote{We use a slightly different notation here than
  in~\cite{FS09}. Our notation $out_{\hA,\hB}^\F$ does not mention the
  name of the input registers and corresponds to $(\F_{\hA,\hB})
  \rho_{UV}$ in~\cite{FS09}.} of the ``ideal-life'' protocol, where
Alice and Bob forward their inputs to $\F$ and output whatever they
obtain from~$\F$.  And we write $out_{\hA,\dhB}^\F$ for the joint
output of the execution of this protocol with a dishonest Bob with
strategy $\dhB$ (and similarly for a dishonest Alice). Note that Bob's possibilities in the ideal world are
very limited: he can produce some classical input $V$ for $\F$ from
his input quantum state $V'$, and then he can prepare and output a
quantum state $Y'$ which might depend on $\F$'s classical reply $Y$.

\subsection{Information-Theoretic Security} \label{sec:infSecurity} 

We define information-theoretic security using the real/ideal-world
paradigm, which requires that by attacking a protocol in the
real world the dishonest party cannot achieve (significantly) more
than when attacking the corresponding functionality in the
ideal world.  To be consistent with the framework used in~\cite{FS09},
we restrict the joint input state, consisting of a classical input to
the honest party and a possibly quantum input to the dishonest party,
to a special form: in case of a dishonest Bob (and correspondingly for
a dishonest Alice), we require that Bob's input consists of a
classical part $Z$ and a quantum part $V'$, such that the joint state
$\rho_{UZV'}$ satisfies $\rho_{UZV'} = \rho_{\MC{U}{Z}{V'}}$, i.e.,
that $V'$ is correlated with Alice's input only via the classical~$Z$.
We call a joint input state of that form (respectively of the form
$\rho_{U'ZV} = \rho_{\MC{U'}{Z}{V}}$ in case of dishonest Alice) a
{\em legitimate} input state.  As shown in~\cite{FS09}, this
restriction on the input state leads to a meaningful security
definition with a composition theorem that guarantees sequential
composition within {\em classical} outer protocols.  Furthermore, the
results of Section~\ref{sec:improving} also hold when quantifying over all
(possibly non-legitimate) joint input states.

\begin{definition}[Unconditional security 
against dishonest Alice] \label{def:unboundedAliceNiceOrder} 
A protocol $\Pi =
  (\A,\B)$ implements an ideal classical functionality $\F$
  unconditionally securely against dishonest Alice, if for any
  real-world adversary $\dA$ there exists an ideal-world adversary
  $\dhA$ such that for any legitimate input state, it
  holds that the outputs in the real and ideal world are statistically
  indistinguishable, i.e.~
$$
out_{\dA,\B}^\Pi \approxs out_{\dhA,\hB}^\F \, .
$$
\end{definition}

\begin{definition}[Unconditional security 
against dishonest Bob] \label{def:qmemoryBob} 
A protocol $\Pi = (\A,\B)$ implements an
  ideal classical functionality $\F$ unconditionally securely against
  dishonest Bob, if for any real-world adversary $\dB$ there exists an
  ideal-world adversary $\dhB$ such that for any legitimate input state, it holds that the outputs in the real and ideal
  world are statistically indistinguishable, i.e.~
$$
out_{\A,\dB}^\Pi \approxs out_{\hA,\dhB}^\F \, .
$$
\end{definition}

It has been shown in Theorem~5.1 in~\cite{FS09} that protocols
fulfilling the above definitions compose sequentially as follows. For
a classical real-life protocol $\Sigma$ which makes at most $k$ oracle
calls to functionalities $\F_1,\ldots,\F_k$, it is guaranteed that
whatever output $\Sigma$ produces, the output produced when the oracle
calls are replaced by $\eps$-secure protocols is at distance at most
$O(k \eps)$.

Notice that in the definitions above, we do \emph{not} require the
running time of ideal-world adversaries to be polynomial whenever the
real-life adversaries run in polynomial time. This way of defining
unconditional security can lead to the (unwanted) effect that
unconditional security does not necessarily imply computational
security. However, for the security of the construction proposed in
this paper, efficient ideal-life adversaries can be guaranteed, as discussed in Section~\ref{sec:efficientsimulation}.

\subsection{Computational Security in the CRS Model} \label{sec:compSecurity}

One can define security against a computationally bounded dishonest
Bob analogously to information-theoretic security with the two
differences that the input given to the parties has to be sampled by
an efficient quantum algorithm and that the output states should be
computationally indistinguishable.

In the common-reference-string (CRS) model, all participants in the
real-life protocol $\Pi_{\A,\B}$ have access to a classical public
string $\crs$ which is chosen before any interaction starts according
to a distribution only depending on the security parameter. On the
other hand, the participants in the ``ideal-life'' protocol
$\F_{\hA,\hB}$ interacting only with the ideal functionality do not make
use of the string~$\crs$. Hence, an ideal-world adversary $\dhB$, that
operates by simulating the real world to the adversary $\dB$, is free to choose
$\crs$ in any way he wishes.

In order to define computational security against a dishonest Bob in
the CRS model, we consider a polynomial-size
quantum circuit, called \emph{input sampler}, which takes as input the
security parameter $m$ and the CRS $\crs$ (chosen according to its
distribution) and produces the input state $\rho_{U Z V'}$; $U$
is Alice's classical input to the protocol, and $Z$ and $V'$ denote
the respective classical and quantum information given to dishonest
Bob. We call the input sampler {\em legitimate} if $\rho_{U ZV'} =
\rho_{\MC{U}{Z}{V'}}$. 

In the following and throughout the article, 
we let $\dBobPoly$ be the family of all {\em polynomial-time} 
quantum strategies for dishonest Bob $\dB$. 

\begin{definition}[Computational security 
against dishonest Bob] \label{def:polyboundedBobCRS} 
A protocol $\Pi = (\A,\B)$
  implements an ideal classical functionality $\F$ computationally
  securely against dishonest Bob, if for any real-world adversary $\dB
  \in \dBobPoly$ who has access to the common reference string $\crs$,
  there exists an ideal-world adversary $\dhB \in \dBobPoly$ not using
  $\crs$ such that for any efficient legitimate input sampler, it holds
  that the outputs in the real and ideal world are
  q-indistinguishable, i.e.~
$$
out_{\A,\dB}^\Pi \approxq out_{\hA,\dhB}^\F \, .
$$
\end{definition}
In Appendix~\ref{app:composition}, we show that also the computational security definition, as given here, allows for (sequential) composition of quantum protocols into classical outer protocols. 

\section{Improving the Security via Commit-and-Open}
\label{sec:improving}
\subsection{Security against Benign Bob} 
\label{sec:defbenign}

In this paper, we consider quantum two-party
protocols that follow a particular but very typical construction design. These protocols consist of two phases, called {\em preparation}
and {\em post-processing} phase, and are as specified in Figure~\ref{fig:BB84-type}. We call a protocol that follows this construction design a {\em BB84-type} protocol. 

\begin{figure}
\begin{framed}
\noindent\hspace{-1.5ex} {\sc Protocol $\Pi$ } \\[-4ex]
\begin{description}\setlength{\parskip}{0.5ex}
\item[{\it Preparation:}] $\A$ chooses $x \in_R \set{0,1}^n$ and $\theta \in_R \set{+,\x}^n$ and sends $\ket{x}_{\theta}$ to~$\B$, and $\B$ chooses $\hat{\theta} \in_R \set{+,\times}^n$ and obtains $\hat{x} \in \set{0,1}^n$ by measuring $\ket{x}_{\theta}$ in basis $\hat{\theta}$. 
\item[{\it Post-processing:}] Arbitrary classical communication and 
local computations. 
\end{description}
\vspace{-1.5ex}
\end{framed}
\vspace{-1.5ex}
 \caption{Generic BB84-type quantum protocol $\Pi$. } 
 \label{fig:BB84-type} 
\end{figure}

The following definition captures information-theoretic
security against a somewhat mildly dishonest Bob who we call a {\em benign} (dishonest) Bob. Such a dishonest Bob is benign in that, in the preparation phase, he does not deviate too much from what he is supposed to do; in the post-processing phase though, he may be arbitrarily dishonest. 

To make this description formal, we fix an arbitrary choice of $\theta$ and an arbitrary value for the classical information, $z$, which $\dB$ may obtain as a result of the preparation phase (i.e.~$z = (\hat{\theta},\hat{x})$ in case $\dB$ is actually honest). 
Let $X$ denote the random variable describing the bit-string $x$, where we understand the distribution $P_X$ of $X$ to be conditioned on the fixed values of $\theta$ and~$z$. 
Furthermore, let $\rho_E$ be the state of $\dB$'s quantum register $E$ after the preparation phase. Note that, still with fixed $\theta$ and~$z$, $\rho_E$ is of the form $\rho_E = \sum_x P_X(x) \rho^x_E$, where $\rho^x_E$ 
is the state of $\dB$'s quantum register in case $X$ takes on the
value $x$.  
In general, the $\rho^x_E$'s may be mixed, but we can think of them as being reduced pure states: $\rho^x_E = \tr_R(\proj{\psi_{ER}^x})$ for a suitable register $R$ and pure states $\ket{\psi_{ER}^x}$; we then call the state $\rho_{ER} = \sum_x P_X(x) \proj{\psi_{ER}^x}$ a {\em pointwise purification} (with respect to $X$) of $\rho_E$. 

Obviously, in case $\dB$ is honest, $X_i$ is fully random
whenever $\theta_i \neq \hat{\theta}_i$, so that $\Hmin\bigl(X|_I
\,\big|\,X|_{\bar{I}} = x|_{\bar{I}}\bigr) =
d_H\bigl(\theta|_I,\hat{\theta}|_I\bigr)$ for every $I \subseteq
\set{1,\ldots,n}$ and every $x|_I$, and $\dB$ does not store any
non-trivial quantum state so that $R$ is ``empty'' and $\Hmax(\rho_{ER}) = \Hmax(\rho_E) = 0$.
A benign Bob $\dB$ is now specified to behave close-to-honestly in the
preparation phase: he produces an auxiliary output $\hat{\theta}$
after the preparation phase, and given this output, we are in a
certain sense close to the ideal situation where Bob really measured
in basis $\hat{\theta}$ as far as the values of $\Hmin\bigl(X|_I
\,\big|\,X|_{\bar{I}} = x|_{\bar{I}}\bigr)$ and
$\Hmax(\rho_{ER})$ are concerned.%
\footnote{The reason why we consider the {\em pointwise purification}
  of $\rho_E$ is to prevent Bob from artificially blowing up
  $\Hmax(\rho_{ER})$ by locally generating a large mixture or storing
  an unrelated mixed input state. }
We now make this precise:

\begin{definition}[Unconditional security 
against {\em benign} Bob] \label{def:BenignBob}
A BB84-type quantum protocol $\Pi$ securely implements $\cal F$ against a {\em $\beta$-benign} Bob for some parameter $\beta \geq 0$, if it securely implements $\cal F$ according to Definition~\ref{def:qmemoryBob}, with the following two modifications:
\begin{enumerate}
\item The quantification is over all $\dB$ with the following
  property: after the preparation phase $\dB$ either aborts, or else
  produces an auxiliary output $\hat{\theta} \in
  \set{+,\x}^n$. Moreover, the joint state of $\A, \dB$ (after
  $\hat{\theta}$ has been output) is statistically indistinguishable
  from a state for which it holds that for any fixed values of
  $\theta$, $\hat{\theta}$ and $z$, for any subset $I \subseteq
  \set{1,\ldots,n}$, and for any $x|_{\bar{I}}$
\begin{equation} \label{eq:staterequirements}
\Hmin\bigl(X|_I \,\big|\,X|_{\bar{I}} = x|_{\bar{I}}\bigr) \geq d_H\bigl(\theta|_I,\hat{\theta}|_I\bigr) - \beta n
\qquad\text{and}\qquad
\Hmax\bigl(\rho_{ER}\bigr) \leq \beta n 
\end{equation}
where $\rho_{ER}$ is the pointwise purification of $\rho_E$ with respect to $X$. 
\item $\dhB$'s running-time is polynomial in the running-time of $\dB$. 
\end{enumerate} 
\end{definition}

\subsection{From Benign to Computational Security}

We show a generic compiler which transforms any BB84-type protocol into a new quantum protocol for the same task. The compiler achieves that if the original protocol is unconditionally secure against dishonest Alice and unconditionally secure against {\em benign} Bob, then the compiled protocol is still unconditionally secure against dishonest Alice and it is {\em computationally secure} against {\em arbitrary} dishonest Bob. 

The idea behind the construction of the compiler is to incorporate a commitment scheme and force Bob to behave benignly by means of a commit-and-open procedure. 
Figure~\ref{fig:compiled} shows the compilation of an arbitrary BB84-type protocol~$\Pi$. The quantum communication is increased from $n$ to $m = n/(1-\alpha)$ qubits, where $0 < \alpha < 1$ is some additional parameter that can be arbitrarily chosen. The compiled protocol also requires 3 more rounds of interaction. 

\begin{figure}
\begin{framed}
\noindent\hspace{-1.5ex} {\sc Protocol $\compile(\Pi)$ } \\[-4ex]
\begin{description}\setlength{\parskip}{0.5ex}
\item[{\it Preparation:}] $\A$ chooses $x \in_R \set{0,1}^m$ and
  $\theta \in_R \set{+,\x}^m$ and sends $\ket{x}_{\theta}$
  to~$\B$. Then, $\B$ chooses $\hat{\theta} \in_R \set{+,\times}^m$ and obtains $\hat{x} \in \set{0,1}^m$ by measuring $\ket{x}_{\theta}$ in basis $\hat{\theta}$. 
\item[{\it Verification:}] 
\begin{enumerate}
\item 
$\B$ commits to $\hat{\theta}$ and $\hat{x}$ position-wise:   
$c_i  \assign \commit\bigl((\hat{\theta}_i,\hat{x}_i),r_i\bigr)$ with 
randomness $r_i$ for $i = 1,\ldots,m$. He sends the commitments to $\A$.
\item\label{step:check} $\A$ sends a random test subset $T
  \subset \{1,\ldots,m \}$ of size $\alpha m$. $\B$ opens
  $c_i$ for all $i \in T$, and $\A$ checks that the
  openings were correct and that $x_i = \hat{x}_i$ whenever $\theta_i
  = \hat{\theta}_i$. If all tests are passed, $\A$ accepts,
  otherwise, she rejects and aborts.
\item The tested positions are discarded by both parties: $\A$ and $\B$ restrict $x$ and $\theta$, respectively $\hat{\theta}$ and $\hat{x}$, to 
$i \in \bar{T}$.
\end{enumerate}
\item[{\it Post-processing:}] As in $\Pi$ (with $x, \theta,\hat{x}$ and $\hat{\theta}$ restricted to the positions $i \in \bar{T}$). 
\end{description}
\vspace{-1.5ex}
\end{framed}
\vspace{-1.5ex}
 \caption{Compiled protocol $\compile(\Pi)$.} 
 \label{fig:compiled} 
\end{figure}

We need to specify
what kind of commitment scheme to use. In order to preserve
unconditional security against dishonest Alice, the commitment scheme needs to be unconditionally hiding, and so can at best be computationally binding. 
However, for a plain computationally binding commitment scheme, we do not know how to reduce the computational security of $\compile(\Pi)$ against dishonest Bob to the computational binding property of the commitment scheme.%
\footnote{Classically, this would be done by a rewinding argument, but this fails to work for a quantum Bob. }
Therefore, we use a commitment scheme with additional properties: we require a {\em keyed} commitment scheme $\commit_\pk$, where the corresponding public key $\pk$ is generated by one of two possible key-generation algorithms: $\GH$ or $\GB$. For a key $\pkH$ generated by $\GH$, the commitment scheme $\commit_\pkH$ is unconditionally hiding, whereas the other generator, $\GB$, actually produces a key {\em pair} $(\pkB,\sk)$, so that the secret key $\sk$ allows to efficiently extract $m$ from $\commit_\pkB(m,r)$, and as such $\commit_\pkB$ is unconditionally binding. Furthermore, we require $\pkH$ and $\pkB$
to be computationally indistinguishable, even against quantum attacks. 
We call such a commitment scheme a {\em dual-mode} commitment
scheme.%
\footnote{The notions of dual-mode {\em cryptosystems} and of meaningful/meaningless encryptions, as introduced in~\cite{PVW08} and~\cite{KN08}, are similar in spirit but differ slightly technically. }
As a candidate for implementing such a system, we propose the public-key encryption scheme of Regev~\cite{Regev05}, which is
based on a worst-case lattice assumption and is not known to be breakable even by
(efficient) quantum algorithms. Regev does not explicitly state that the
scheme has the property we need, but this is implicit in his proof that the underlying
computational assumption implies semantic security.%
\footnote{The proof
   compares the case where the public key is generated normally to a
   case where it is chosen with no relation to any secret key. It is
   then argued that the assumption implies that the two cases are
   computationally  indistinguishable, and that in the second case, a
   ciphertext carries essentially no information about the
   message. This argument implies what we need.}

For simplicity and efficiency, we consider the common-reference-string model, and we assume the key $\pkB$ for the commitment scheme,
generated according to $\GB$, to be contained in the CRS. 
We sketch in Section~\ref{sec:CRS} 
how to avoid the CRS model, at the cost of a non constant-round
construction where the parties generate the CRS jointly by means of a
coin-tossing protocol (see~\cite{DL09arxiv} for details).

We sometimes write $\compile_\pkH(\Pi)$ for the compiled protocol $\compile(\Pi)$ to stress that a key $\pkH$ produced by $\GH$ is used for the dual-mode commitment scheme, and we write $\compile_\pkB(\Pi)$ when a key $\pkB$ produced by $\GB$ is used instead. 

\begin{theorem}\label{thm:Compiler}
Let $\Pi$ be a BB84-type protocol, unconditionally secure against dishonest Alice and against $\beta$-benign Bob for some constant $\beta > 0$. 
Consider the compiled protocol $\compile(\Pi)$ for an arbitrary $\alpha > 0$, where the commitment scheme is instantiated by a dual-mode commitment scheme as described above. 
Then, $\compile(\Pi)$ is unconditionally secure against dishonest Alice and computationally secure against dishonest Bob in the CRS model. 
\end{theorem}
We now prove this theorem, which assumes noise-free quantum
communication; we explain in Section~\ref{sec:Noise} how to generalize
it for a noisy quantum channel. Correctness is obvious. In order to
show unconditional security against dishonest Alice, we notice that
the unconditional hiding property of the commitment scheme ensures
that dishonest Alice does not learn any additional
information. Furthermore, as the ideal-life adversary $\dhA$ is not
required to be time-bounded by
Definition~\ref{def:unboundedAliceNiceOrder}, she can break the
binding-property of the commitment scheme and thereby perfectly
simulate the behavior of an honest Bob towards $\hA$ attacking
$\compile(\Pi)$. The issue of efficiency of the ideal-life adversaries
is addressed in Section~\ref{sec:efficientsimulation}.

  As for computational security
against dishonest Bob, according to
Definition~\ref{def:polyboundedBobCRS}, we need to prove that for
every real-world adversary $\dB \in \dBobPoly$ attacking
$\compile(\Pi)$, there exists a suitable ideal-world adversary $\dhB
\in \dBobPoly$ attacking $\F$ such that
$$
out_{\A,\dB}^{\compile(\Pi)} \stackrel{q}{\approx} out_{\hA,\dhB}^\F \, .
$$
First, note that by the computational indistinguishability of $\pkH$ and  $\pkB$,
\begin{equation}\label{eq:KeySwitch}
out_{\A,\dB}^{\compile(\Pi)} = out_{\A,\dB}^{\compile_\pkH(\Pi)} \stackrel{q}{\approx} out_{\A,\dB}^{\compile_{\pkB}(\Pi)} \, .
\end{equation} 
Then, we construct an adversary $\dB_\circ \in \dBobPoly$ who attacks the unconditional security against benign Bob of protocol $\Pi$, and which satisfies 
\begin{equation}\label{eq:BenignBob}
out_{\A,\dB}^{\compile_\pkB(\Pi)} = out_{\A_\circ,\dB_\circ}^{\Pi} \, ,
\end{equation} 
where $\A_\circ$ honestly executes $\Pi$. 
We define $\dB_\circ$ in the following way. Consider the execution of $\compile(\Pi)$ between $\A$ and $\dB$. 
We split $\A$ into two players $\A_\circ$ and $\tilde{\A}$, where we
think of $\tilde{\A}$ as being placed in between $\A_\circ$ and $\dB$,
see Figure~\ref{fig:Players}. $\A_\circ$ plays honest Alice's part of $\Pi$ while $\tilde{\A}$ acts as follows: It receives $n$ qubits from $\A_\circ$, produces $\alpha n/(1-\alpha)$ random BB84 qubits of its own and interleaves them randomly with those received and sends the resulting $m= n/(1-\alpha)$ qubits to $\dB$. It then does the verification step of $\compile(\Pi)$ with $\dB$, asking to have commitments corresponding to its own qubits opened. If this results in accept, it lets $\A_\circ$ finish the protocol with $\dB$. 
Note that the pair $(\A_\circ,\tilde{\A})$ does exactly the same as $\A$; however, we can also move the actions of $\tilde{\A}$ to Bob's side, and define $\dB_\circ$ as follows. 
$\dB_\circ$ samples $(\pkB,\sk)$ according to $\GB$ and executes $\Pi$ with $\A$ by locally running $\tilde{\A}$ and $\dB$, using $\pkB$ as CRS. If $\tilde{\A}$ accepts the verification then $\dB_\circ$ outputs $\hat{\theta} \in \set{0,1}^n$ (as required from a {\em benign} Bob), obtained by decrypting the unopened commitments  with the help of $\sk$; else, $\dB_\circ$ aborts at this point. It is now clear that 
Equation~\eqref{eq:BenignBob} holds: exactly the same computation takes place in both ``experiments'', the only difference being that they are executed partly by different entities. 
The last step is to show that
\begin{equation}\label{eq:ByAssumption}
out_{\A_\circ,\dB_\circ}^{\Pi} \stackrel{s}{\approx} out_{\hA,\dhB}^\F  \, ,
\end{equation} 
for some $\dhB$.
It is clear that the theorem follows from~\eqref{eq:KeySwitch} - \eqref{eq:ByAssumption} together.
\begin{figure}
\begin{center}
\begin{tikzpicture}
   \tikzstyle{every node}=[minimum size=6mm]
   \draw (0,0) node[draw] (zeroA) {$\A_\circ$};
   \draw (5,0) node[draw] (tildeA) {$\tilde{\A} $};
   \draw (10,0) node[draw] (dB) {$\dB$};
   \draw[rounded corners] (zeroA.east) -- node[below] {$\Pi$} (tildeA);
   \draw[rounded corners] (tildeA.east) -- node[above] {$\compile(\Pi)$} (dB);
   \draw[rounded corners] (zeroA.north) ++(-4mm,1mm) --  +(0mm, 2mm) -| node[pos=0.25,above] {$\A$} ([shift={(4mm,1mm)}] tildeA.north);
   \draw[rounded corners] (tildeA.south) ++(-4mm,-1mm) -- +(0mm, -2mm) -| node[pos=0.25,below] {$\dB_\circ$} ([shift={(4mm,-1mm)}] dB.south);
\end{tikzpicture}
\end{center}
\vspace{-3.5ex}
\small
 \caption{Constructing an attacker $\dB_\circ$ against $\Pi$ from an attacker $\dB$ against $\compile(\Pi)$.}\label{fig:Players} 
\end{figure}

Now~\eqref{eq:ByAssumption} actually claims that $\hA,\dhB$ successfully simulate $\A_\circ$ and $\dB_\circ$ executing 
$\Pi$, and this claim follows  by assumption of benign security 
of $\Pi$ if we show that  $\dB_\circ$ is $\beta$-benign according to 
Definition~\ref{def:BenignBob} for any  $\beta > 0$. 
We show this in the following subsection, i.e.,
 the joint state of $\A_\circ, \dB_\circ$ after the preparation phase
 is statistically indistinguishable from a state  $\rho_{Ideal}$ which satisfies 
the bounds \eqref{eq:staterequirements} from Definition~\ref{def:BenignBob}. 

\subsection{Completing the Proof: Bounding Entropy and Memory Size}\label{sec:Bounds}
First  recall that $\A_\circ$
executing $\Pi$ with $\dB_\circ$ can equivalently be thought of as
$\A$ executing $\compile_\pkB(\Pi)$ with $\dB$. Furthermore, a joint
state of $\A,\dB$ is clearly also a joint state of $\A_\circ, \dB_\circ$. 

To show the existence of $\rho_{Ideal}$ as promised above, it therefore suffices to show such a
state for  $\A,\dB$. In other words,
we need to show that  the execution of $\compile_\pkB(\Pi)$ with honest Alice $\A$ and arbitrarily dishonest Bob $\dB$ will, after verification, be close to a state where \eqref{eq:staterequirements} holds. 
To show this closeness, we consider an equivalent EPR-pair version, where Alice creates $m$ EPR pairs $(\ket{00}+\ket{11})/\sqrt{2}$, sends one qubit in 
each pair to Bob and keeps the others in register $A$. Alice measures her qubits only when needed: she measures the qubits within $T$ in Step~\ref{step:check} of the verification phase, and the remaining qubits at the end of the verification phase. 
With respect to the information Alice and Bob obtain, this EPR version
is {\em identical} to the original protocol $\compile_\pkB(\Pi)$: the
only difference is the point in time when Alice obtains certain
information. Furthermore, we can also do the following modification
without affecting~\eqref{eq:staterequirements}. Instead of measuring her qubits
in $T$ in {\em her} basis $\theta|_T$, she measures them in {\em
  Bob's} basis $\hat{\theta}|_T$; however, she still verifies only whether $x_i = \hat{x}_i$ for those $i \in T$ with $\theta_i = \hat{\theta}_i$. 
Because the positions $i \in T$ with $\theta_i \neq \hat{\theta}_i$
are not used in the protocol at all, this change has no effect. As the
commitment scheme is unconditionally binding if key $\pkB$ is used, Bob's basis $\hat{\theta}$ is well defined by his commitments (although hard to compute), even if Bob is dishonest. The resulting scheme is given in Figure~\ref{fig:EPR-Version}. 

\begin{figure}
\begin{framed}
\noindent\hspace{-1.5ex} {\sc Protocol EPR-$\compile_\pkB(\Pi)$ } \\[-4ex]
\begin{description}\setlength{\parskip}{0.5ex}
\item[{\it Preparation:}] $\A$ prepares $m$ EPR pairs and sends the second qubit in each pair to Bob while keeping the others in register $A = A_1\cdots A_m$. 
$\B$ chooses $\hat{\theta} \in_R \set{+,\times}^m$ and obtains $\hat{x} \in \set{0,1}^m$ by measuring the received qubits in basis $\hat{\theta}$. 
\item[{\it Verification:}] 
\begin{enumerate}
\item 
$\B$ commits to $\hat{\theta}$ and $\hat{x}$ position-wise:   
$c_i  \assign \commit\bigl((\hat{\theta}_i,\hat{x}_i),r_i\bigr)$ with 
randomness $r_i$ for $i = 1,\ldots,m$. He sends the commitments to $\A$.
\item\label{step:EPRcheck} $\A$ sends a random test subset $T
  \subset \{1,\ldots,m \}$ of size $\alpha m$. $\B$ opens
  $c_i$ for all $i \in T$.
  $\A$ chooses $\theta \in_R \set{+,\x}^m$, measures registers $A_i$ with $i \in T$ in basis $\hat{\theta}_i$ to obtain $x_i$, and she checks that the
  openings were correct and that $x_i = \hat{x}_i$ whenever $\theta_i
  = \hat{\theta}_i$ for $i \in T$. If all tests are passed, $\A$ accepts,
  otherwise, she rejects and aborts the protocol.
\item $\A$ measures the remaining registers in basis $\theta|_{\bar{T}}$ to obtain $x|_{\bar{T}}$. The tested positions are discarded by both parties: $\A$ and $\B$ restrict $x$ and $\theta$, respectively $\hat{\theta}$ and $\hat{x}$, to the positions $i \in \bar{T}$.
\end{enumerate}
\item[{\it Post-processing:}] As in $\Pi$ (with $x, \theta,\hat{x}$ and $\hat{\theta}$ restricted to the positions $i \in \bar{T}$). 
\end{description}
\vspace{-1.5ex}
\end{framed}
\vspace{-1.5ex}
 \caption{EPR version of $\compile_\pkB(\Pi)$.} 
 \label{fig:EPR-Version} 
\end{figure}

We consider an execution of the scheme from
Figure~\ref{fig:EPR-Version} with an honest Alice $\A$ and a dishonest
Bob $\dB$, and we fix $\hat{\theta}$ and $\hat{x}$, determined by
Bob's commitments.  Let $\ket{\varphi_{AE}} \in {\cal H}_A\otimes{\cal
  H}_E$ be the state of the joint system right before
Step~\ref{step:check} of the verification phase.  Since in the end, we
are anyway interested in the pointwise purification of Bob's state, we
may indeed assume this state to be pure; if it is not, then we purify
it and carry the purifying register $R$ along with $E$.
Clearly, if $\dB$ had honestly done his measurements then $\ket{\varphi_{AE}} = \ket{\hat{x}}_{\hat{\theta}}\otimes \ket{\varphi_E}$ for some $\ket{\varphi_E}\in {\cal H}_E$.
In this case, the quantum memory $E$ would be empty: $H_0(\proj{\varphi_E})=0$. 
Moreover, $X$, obtained by measuring $A|_{\bar{T}}$ in basis $\theta|_{\bar{T}}$, would contain $d_H(\theta|_{\bar{T}},\hat{\theta}|_{\bar{T}})$ random bits. 
We show that the verification phase enforces these properties, at least approximately in the sense of~\eqref{eq:staterequirements}, for an arbitrary dishonest Bob $\dB$. 

In the following, $r_H(\cdot,\cdot)$ denotes the relative Hamming distance between two strings, i.e., the Hamming distance divided by their length. 
Recall that $T \subset \set{1,\ldots,m}$ is random subject to $|T| = \alpha m$. Furthermore, for a fixed $\hat{\theta}$ but a randomly chosen $\theta$, the subset $T' = \Set{i \in T}{ \theta_i = \hat{\theta}_i}$ is a random subset (of arbitrary size) of $T$. 
Let the random variable $\Test$ describe the choice of $\test = (T,T')$ as specified above, and consider the state
$$
\rho_{\Test AE} = \rho_{\Test} \otimes \proj{\varphi_{AE}} = \sum_{\test} P_{\Test}(\test) \proj{\test} \otimes \proj{\varphi_{AE}}
$$
consisting of the classical $\Test$ and the quantum state $\ket{\varphi_{AE}}$. 

\begin{lemma}\label{lemma:idealstate}
For any $\eps > 0$, $\hat{x}\in\{0,1\}^m$ and 
$\hat{\theta}\in \{+,\times\}^m$, the state $\rho_{\Test AE}$ is 
negligibly close (in $m$) to a state 
$$
\tilde{\rho}_{\Test AE} = \sum_{\test} P_{\Test}(\test) \proj{\test} \otimes \proj{\tilde{\varphi}^{\test}_{AE}}
$$
where for any $\test = (T,T')$: 
$$
\ket{\tilde{\varphi}^{\test}_{AE}} = \sum_{x \in B_{\test}} 
\alpha^{\test}_x \ket{x}_{\hat{\theta}} \ket{\psi_E^x} 
$$
for $B_{\test} = \{x \in \{0,1\}^m \,|\, r_H(x|_{\bar{T}},\hat{x}|_{\bar{T}}) \leq
r_H(x|_{T'},\hat{x}|_{T'}) + \eps \}$ and arbitrary coefficients $\alpha^{\test}_x \in \C$. 
\end{lemma}
In other words, we are close to a situation where for {\em any} choice
of $T$ and $T'$ and for {\em any} outcome $x|_T$ when measuring $A|_T$
in basis $\hat{\theta}|_T$, the relative error
$r_H(x|_{T'},\hat{x}|_{T'})$ gives an upper bound (which holds with
probability 1) on the relative error
$r_H(x|_{\bar{T}},\hat{x}|_{\bar{T}})$ one would obtain by measuring
the remaining subsystems $A_i$ with $i \in \bar{T}$ in basis
$\hat{\theta}_i$. 
\begin{proof}
For any $\test$ we let $\ket{\tilde{\varphi}^{\test}_{AE}}$ be the
  renormalized projection of $\ket{\varphi_{AE}}$ into the
  subspace $\mathrm{span}\{\ket{x}_{\hat{\theta}} \,|\, x \in
  B_{\test} \} \otimes {\cal H}_E$ and let
  $\ket{\tilde{\varphi}^{\test \perp}_{AE}}$ be the renormalized
  projection of $\ket{\varphi_{AE}}$ into the orthogonal
  complement, such that $\ket{\varphi_{AE}} = \eps_{\test}
  \ket{\tilde{\varphi}^{\test}_{AE}} + \eps_{\test}^\perp
  \ket{\tilde{\varphi}^{\test \perp}_{AE}}$ with $\eps_{\test} =
  \braket{\tilde{\varphi}^{\test}_{AE}}{\varphi_{AE}}$ and
  $\eps_{\test}^\perp = \braket{\tilde{\varphi}^{\test
      \perp}_{AE}}{\varphi_{AE}}$. By construction,
  $\ket{\tilde{\varphi}^{\test}_{AE}}$ is of the form
  required in the statement of the lemma. 
  A basic property of the trace norm of pure
  states gives
$$
\delta\bigl( \proj{\varphi_{AE}}, \proj{\tilde{\varphi}^{\test}_{AE}}
\bigr) = \sqrt{1 -
  |\braket{\tilde{\varphi}^{\test}_{AE}}{\varphi_{AE}}|^2} =
|\eps_{\test}^\perp| \, .
$$
This last term corresponds to the 
square root of the probability, when given $\test$, 
to observe a string $x \not\in B_{\test}$ when measuring subsystem $A$ of 
$\ket{\varphi_{AE}}$ in basis~$\hat{\theta}$. Furthermore, using elementary 
properties of the trace norm and Jensen's inequality gives 
\begin{align*}
\delta\bigl( \rho_{\Test AE},\tilde{\rho}_{\Test AE} \bigr)^2 &= \bigg( \sum_{\test} P_{\Test}(test) \,\delta\bigl( \proj{\varphi_{AE}}, \proj{\tilde{\varphi}^{\test}_{AE}} \bigr) \bigg)^2 \\
&= \bigg( \sum_{\test} P_{\Test}(test) \, |\eps_{\test}^\perp| \bigg)^2 
\leq \sum_{\test} P_{\Test}(test) \, |\eps_{\test}^\perp|^2 \, , 
\end{align*}
where the last term is the probability to observe a string $x \not\in B_{\test}$ when choosing $\test$ according to $P_{\Test}$ and measuring subsystem $A$ of $\ket{\varphi_{AE}}$ in basis $\hat{\theta}$. This situation, though, is a classical sampling problem, for which it is well known that for any measurement outcome $x$, the probability (over the choice of $\test$) that $x \not\in B_{\test}$ is negligible in $m$ (see~e.g.~\cite{Hoeffding63}). 
\qed
\end{proof}
In combination with Lemma~\ref{lemma:SmallSuperpos} on ``small
superpositions of product states'', and writing $h$ for the binary
entropy function $h(\mu) = -\big(\mu\log(\mu)+(1-\mu)\log(1-\mu)\big)$
as well as using that $\big|\{ y \in \{0,1\}^n \,|\, d_H(y,\hat{y})
\leq \mu n \}\big| \leq 2^{h(\mu) n}$ for any $\hat{y} \in
\set{0,1}^n$ and $0 \leq \mu \leq \frac12$, we can conclude the
following.

\begin{corollary}
\label{corSerge}
Let 
$\tilde{\rho}_{\Test AE}$ be of the form as in Lemma~\ref{lemma:idealstate} (for given $\eps$, $\hat{x}$ and $\hat{\theta}$). 
For any fixed $\test = (T,T')$ and for any fixed $x|_T \in \{0,1\}^{\alpha m}$ with \smash{$\err \assign r_H(x|_{T'},\hat{x}|_{T'}) \leq \frac12$}, let $\ket{\psi_{AE}}$ be the state to which \smash{$\ket{\tilde{\varphi}^{\test}_{AE}}$} collapses when for every $i \in T$ subsystem $A_i$ is measured in basis $\hat{\theta}_i$ and $x_i$ is observed, where we understand $A$ in $\ket{\psi_{AE}}$ to be restricted to the registers $A_i$ with $i \in \bar T$. 
Finally, let $\sigma_E = \tr_A(\proj{\psi_{AE}})$ and let the random variable $X$ describe the outcome when measuring the remaining $n = (1-\alpha) m$ subsystems of $A$ in basis $\theta|_{\bar{T}} \in \set{+,\times}^n$. 
Then, for any subset $I \subseteq \set{1,\ldots,n}$ and any $x|_I$,%
\footnote{Below, $\theta|_I$ (and similarly $\hat{\theta}|_I$) should be understood as first restricting the $m$-bit vector $\theta$ to $\bar{T}$, and then restricting the resulting $n$-bit vector $\theta|_{\bar{T}}$ to $I$: $\theta|_I\assign {(\theta|_{\bar{T}})|}_I$.  }
$$
\Hmin\bigl(X|_I \,\big|\,X|_{\bar{I}} = x|_{\bar{I}}\bigr) \geq d_H\bigl(\theta|_I,\hat{\theta}|_I\bigr) - h(\err + \eps) n
\quad\text{and}\quad
\Hmax\bigl(\sigma_E\bigr) \leq h(\err + \eps) n \, .
$$
\end{corollary}
Thus, the number of errors between the measured $x|_{T'}$ and the given $\hat{x}|_{T'}$ gives us a bound on the min-entropy of the outcome when measuring the remaining subsystems of $A$, and on the max-entropy of the state of subsystem $E$. 

 \begin{proof}
 To simplify notation, we write $\vartheta = \theta|_{\bar{T}}$ and $\hat{\vartheta} = \hat{\theta}|_{\bar{T}}$.  
 By definition of $\tilde{\rho}_{\Test AE}$, for any fixed values of $\eps,\hat{x}$, and $\hat{\theta}$, the state $\ket{\psi_{AE}}$ is of the form  $\ket{\psi_{AE}} = \sum_{y \in {\cal Y}} \alpha_y
 \ket{y}_{\hat{\vartheta}} \otimes \ket{\psi^y_E}$, where ${\cal Y} = \Set{y \in \set{0,1}^n}{d_H(y,\hat{x}|_{\bar{T}}) \leq \err + \eps}$. 
 Consider the corresponding mixture $\tilde{\sigma}_{AE} = \sum_{y \in {\cal Y}} 
 |\alpha_y|^2 \ket{y}_{\hat{\vartheta}} \bra{y}_{\hat{\vartheta}}
 \otimes\proj{\psi_E^y}$ and define $\tilde{X}$ as the random variable 
 for the outcome when measuring register $A$ of $\tilde{\sigma}_{AE}$ 
 in basis $\vartheta$. Notice that 
 $H_{\infty}(\tilde{X})\geq d_H(\vartheta,\hat{\vartheta})$ 
 since any state $\ket{y}_{\hat{\vartheta}}$, when measured in basis $\vartheta$, 
 produces a random bit for every position $i$ with $\vartheta \neq \hat{\vartheta}$. 
 Lemma~\ref{lemma:SmallSuperpos} allows us to conclude that 
 $H_{\infty}(X) \geq H_{\infty}(\tilde{X})-\log{|{\cal Y}|} \geq  
 d_H(\vartheta,\hat{\vartheta})-h(err+\eps)n$ and
 $H_0(\sigma_E)\leq \log |{\cal Y}| \leq h(err+\eps)n$.  
 This proves the claim for $I = \set{1,\ldots,n}$. For arbitrary $I \subset \set{1,\ldots,n}$ and $x|_I$, we can consider the pure state obtained by measuring the registers $A_i$ with $i \not\in I$ in basis $\vartheta_i$ when $x|_{\bar{I}}$ is observed. This state is still a superposition of at most $|{\cal Y}|$ vectors and thus we can apply the exact same reasoning to obtain \eqref{eq:staterequirements}. 
 \qed
 \end{proof}

The claim to be shown now follows by combining Lemma~\ref{lemma:idealstate} and Corollary~\ref{corSerge}. 
Indeed, the ideal state $\rho_{Ideal}$ we promised is produced by putting $\A$ and $\dB$ in the state $\tilde{\rho}_ {Test A E}$ defined in Lemma~\ref{lemma:idealstate}, and running Steps 2 and 3 of the verification phase. This state is negligibly close to the real state since by Lemma~\ref{lemma:idealstate} we were negligibly close to the real state before these operations. 
Corollary~\ref{corSerge} guarantees that \eqref{eq:staterequirements} is satisfied.

\section{Extensions and Generalizations}

 \subsection{In the Presence of Noise} \label{sec:Noise}
 In the description of the compiler $\compile$ and in its analysis, we
 assumed the quantum communication to be noise-free. Indeed, if the
 quantum communication is noisy honest Alice is likely to reject an
 execution with honest Bob. It is straightforward to generalize the
 result to noisy quantum communication: In Step~\ref{step:check} in the
 verification phase of $\compile(\Pi)$, Alice rejects and aborts if the
 relative number of errors between $x_i$ and $\hat{x}_i$ for $i \in T$
 with $\theta_i = \hat{\theta}_i$ exceeds the error probability $\phi$
 induced by the noise in the quantum communication by some small $\eps'
 > 0$. By Hoeffding's inequality~\cite{Hoeffding63}, this guarantees
 that honest Alice does not reject honest Bob except with exponentially
 small probability. Furthermore, proving the security of this
 ``noise-resistant'' compiler goes along the exact same lines as for
 the original compiler. The only difference is that when applying
 Corollary~\ref{corSerge}, the parameter $\err$ has to be chosen as
 $\err = \phi + \eps'$, so that \eqref{eq:staterequirements} holds for
 $\beta = h(\err+\eps) = h(\phi+\eps'+\eps)$ and thus the claim of
 Theorem~\ref{thm:Compiler} hold for any $\beta > h(\phi)$ (by choosing
 $\eps,\eps' > 0$ small enough).  This allows us to generalize the
 results from the Section~\ref{sec:Examples} to the setting of noisy
 quantum communication.

\subsection{Bounded-Quantum-Storage Security}\label{sec:HybridSecurity}
In this section we show that our compiler preserves security in the bounded-quantum-storage model (BQSM). In this model, one of the players (Bob in our case) is assumed be able to store only a limited number of qubits beyond a certain point in the protocol. BQSM-secure OT and identification protocols are known \cite{DFRSS07,DFSS07}, but they can be efficiently broken if the memory bound does not hold. Therefore, by the theorem below, applying the compiler produces protocols with better security, namely the adversary needs large quantum storage {\em and} large computing power to succeed.

Consider a  BB84-type protocol $\Pi$, and for a constant $0 < \gamma < 1$, let $\dBobBQSM^\gamma(\Pi)$ be the set of dishonest players $\dB$ that store only $\gamma n$ qubits after a certain point in $\Pi$, where $n$ is the number of qubits sent initially. 
Protocol $\Pi$ is said to be unconditionally secure against $\gamma$-BQSM Bob, if it satisfies Definition~\ref{def:qmemoryBob} with the restriction that the quantification is over all dishonest $\dB \in \dBobBQSM^\gamma(\Pi)$. 

\begin{theorem}\label{thm:BQSM}
  If $\Pi$ is unconditionally secure against $\gamma$-BQSM Bob, then
  $\compile(\Pi)$ (for an $0<\alpha<1$) is unconditionally secure
  against $\gamma(1\!-\!\alpha)$-BQSM Bob.
\end{theorem}

\begin{proof}
  Exactly as in the proof of Theorem~\ref{thm:Compiler}, given
  dishonest Bob $\dB$ attacking $\compile(\Pi)$, we construct
  dishonest Bob $\dB_\circ$ attacking the original protocol $\Pi$. The
  only difference here is that we let $\dB_\circ$ generate the CRS
  ``correctly'' as $\pkH$ sampled according to $\GH$. It follows by
  construction of $\dB_\circ$ that $out^{\compile(\Pi)}_{\A,\dB} =
  out^{\Pi}_{\A_\circ, \dB_\circ} \, .$ Also, it follows by
  construction of $\dB_\circ$ that if $\dB \in
  \dBobBQSM^{\gamma(1-\alpha)}(\compile(\Pi))$ then $\dB_\circ \in
  \dBobBQSM^\gamma(\Pi)$, since $\dB_\circ$ requires the same amount
  of quantum storage as $\dB$ but communicates an $\alpha$-fraction
  fewer qubits. It thus follows that there exists $\dhB$ such that
  $out^{\Pi}_{\A_\circ, \dB_\circ} \stackrel{s}{\approx} out^{\cal
    F}_{\hA,\dhB} \, .$ This proves the claim.  \qed
\end{proof}

\subsection{Efficient Simulation} \label{sec:efficientsimulation}
The security definitions we use here are clearly closely related to the UC-security
concept in that they require a protocol to implement a certain functionality, and that this can be demonstrated via a simulation argument. However, our definitions do not imply UC-security.
For this we would need all simulators to be efficient, and 
our definition of unconditional security against dishonest Alice does not require this (unlike the definition of computational security against Bob).

Of course, it might still be the case that our compilation preserves efficiency of the simulator, namely if protocol $\Pi$ is secure against dishonest Alice with efficient simulator $\dhA$, 
then so is $\compile(\Pi)$. 

Although this would be desirable, it does not seem to be the case for our basic construction: In order to show such a result, we would need to simulate the preprocessing phase against dishonest $\dA$ efficiently and without measuring the qubits that are not  ``opened during'' preprocessing. Once this is done, we can give the remaining qubits to $\dhA$ who can simulate the rest of the protocol.

However, the whole point of the preprocessing is to ensure that Bob measures all qubits, unless he can break the binding property of the commitments, so the only hope is to bring the simulator in a situation where it can make commitments and open them any way it wants. The standard way to do this is to give the simulator some trapdoor information related to the common reference string, that Bob would not have in real life. Indeed, with such a trapdoor commitment scheme, simulation of the preprocessing is trivial: We just wait until Alice reveals
the bases and the test subset, measure qubits in the test subset, and open the commitments according to the measurement results.

While no such trapdoor is known for the commitment scheme we suggested earlier, it is possible to extend the construction efficiently to build in such a trapdoor:

To do this, we need a new ingredient, namely a relation $R$ representing a hard problem, and a 
$\Sigma$-protocol for $R$.  The relation is a set of pairs $R= \{ (u,w) \}$ where $u$ can be thought of as a problem instance and $w$ as the solution. The relation is hard if one can efficiently generate 
$(u,w) \in R$ such that from $u$ one cannot in polynomial time compute $w$ such that $(u,w)\in R$.
We need that $R$ is hard even for quantum algorithms. We also need that there is a $\Sigma$-protocol, i.e. an honest verifier perfect zero-knowledge interactive proof of knowledge where a prover shows, on input $u$, that he knows $w$ such that $(u,w)\in R$. Protocol conversations have form
$(a,b,z)$ where the prover sends $a$, the verifier gives a random challenge bit $b$ and the prover sends $z$. It is required that, given conversations $(a,0,z_0), (a,1,z_1)$ that the verifier would accept, one can compute $w$ such that $(u,w)\in R$.

As an example, one can think of  $u= (G_0,G_1)$ where $G_0,G_1$ are isomorphic graphs and $w$ is an isomorphism. The $\Sigma$- protocol is just the well-known standard zero-knowledge proof for graph isomorphism. There are several plausible and practically more useful examples, see \cite{DFS04}. 

Given this, and a commitment scheme with public key $\pkH$ as described above, we build a new commitment scheme as follows: the public key is $u, \pkH$. To commit to a bit $b$, the committer runs the honest verifier simulator to get a conversation $(a,b,z)$. The commitment is now
$a, c_0,c_1$, where $c_b= \commit(z,r)$ and $c_{1-b} = \commit(0,r')$.
To open a commitment, one reveals $b$ and opens $c_b$. The receiver checks that $(a,b,z)$ is accepting and that $c_b$ was correctly opened.

By perfect honest verifier zero-knowledge and perfect hiding of commitments based on $\pkH$, 
the new commitment is perfectly hiding. However, if one knows $w$ such that $(u,w)\in R$, one can 
compute $(a,0,z_0)$ and  $(a,1,z_1)$ both of which are accepting conversations, and set
$c_0= \commit(z_0,r_0)$, $c_1= \commit(z_1,r_1)$, and it is now possible to open both ways. 
Hence $w$ serves as the trapdoor we need for efficient simulation above. 

The new commit scheme still has the property we need for the compilation, namely one can choose the public key in a different but indistinguishable way, such that the committed bit can be extracted: we let the public key be $u,\pkB$, where $\pkB$ is a binding public key for our original scheme. Now, given a commitment
$(a,c_0,c_1)$, we can decrypt $c_0,c_1$ to see which of them contains a valid reply in the $\Sigma$-protocol. The only way we can fail to predict how the commitment can be opened is if both $c_0$ and $c_1$ contain valid replies. But this would imply that the committer can compute $w$, so for a polynomial-time bounded committer, this only happens with negligible probability, since the relation is assumed to be hard.

\subsection{Doing without a Common Reference String} 
\label{sec:CRS}
We can get rid of the CRS assumption by instead generating a reference string from scratch
using a coin-flip protocol. In~\cite{DL09arxiv}, such a coin-flip protocol
is described and proved secure against quantum adversaries using
Watrous' quantum rewinding method~\cite{Watrous06}.
Note that for our compiler, we want the CRS to be an
unconditionally hiding public key, and when using Regev's
cryptosystem, a uniformly random string (as output by the coin-flip)
does indeed determine such a key, except with negligible probability.

\section{Applications}\label{sec:Examples}

\subsection{Oblivious Transfer}
 \label{sec:QOT}
 \newcommand{\otm}[1]{1-2 OT$^{#1}$}
\newcommand{\fot}{\ensuremath{{\cal F}_{\OT{\ell}}}}
 We discuss a protocol that securely implements one-out-of-two
 oblivious transfer of strings of length~$\ell$ (i.e. \otm{\ell}). 
 In \otm{\ell},
 the sender $\A$ sends two $l$-bit strings $s_0$ and $s_1$
 to the receiver $\B$. $\B$ can choose which string to receive ($s_k$)
 but does not learn anything about the other one ($s_{1-k}$). On the
 other hand, $\A$ does not learn $\B$'s choice bit $k$.
 The protocol is almost identical 
 to the \otm{1} introduced in~\cite{BBCS91}, but uses hash functions instead of parity
 values to mask the inputs $s_0$ and $s_1$. The resulting scheme, called $\QOT$, is presented
 in Figure~\ref{fig:ot}, where $\cal F$ denotes a suitable family of universal hash functions with range $\set{0,1}^\ell$ (as specified in~\cite{DFRSS07}). We assume that $\ell = \lfloor \lambda n \rfloor$ for some constant $\lambda > 0$. 

 \begin{figure}
 \begin{framed}
 \noindent\hspace{-1.5ex} {\sc Protocol $\QOT:$ } \\[-4ex]
 \begin{description}\setlength{\parskip}{0.5ex}
 \item[{\it Preparation:}] $\A$ chooses $x \in_R \set{0,1}^n$ and 
 $\theta \in_R \set{+,\x}^n$ and sends $\ket{x}_{\theta}$ to~$\B$, and $\B$ 
 chooses $\hat{\theta} \in_R \set{0,1}^n$ and obtains $\hat{x} \in \set{0,1}^n$ 
 by measuring $\ket{x}_{\theta}$ in basis $\hat{\theta}$. 
 \item[{\it Post-processing:}]\begin{enumerate}
 \item
 $\A$ sends $\theta$ to $\B$. 
 \item
 $\B$ partitions all positions $1\leq i \leq n$ in 
 two subsets according to his choice bit $k \in \{ 0,1 \}$: the ``good'' subset $I_k \assign \{ i: \theta_i = \hat{\theta}_i \}$ and the ``bad'' subset $I_{1-k} \assign \{ i: \theta_i \neq \hat{\theta}_i \}$. $\B$ sends $(I_0,I_1)$
 to $A$.
 \item
 $\A$ sends descriptions of
  $f_0,f_1\in_R {\cal F}$ 
 together with 
 $m_0 \assign s_0 \oplus f_0(x|_{I_0})$ and $m_1 \assign s_1 \oplus f_1(x|_{I_1})$.
 \item 
 $\B$ computes $s_k = m_k \oplus f_k(\hat{x}|_{I_k})$.
 \end{enumerate}
 \end{description}
 \vspace{-1.5ex}
 \end{framed}
 \vspace{-1.5ex}
  \caption{Protocol for String OT.}
  \label{fig:ot} 
 \end{figure}

 \begin{theorem}\label{thm:OT}
 Protocol $\QOT$ is unconditionally secure against $\beta$-benign Bob for any $\beta < \frac18 - \frac{\lambda}{2}$. 
 \end{theorem}
 \begin{proof}
   Let $\fot$ be the ideal oblivious transfer functionality. For any
   given benign Bob $\dB$, we construct $\dhB$ the following way.
   $\dhB$ runs locally a copy of $\dB$ and simulates Alice by running
   $\A$ up to but not including Step~3.  After the preparation phase,
   $\dhB$ gets $\hat{\theta}$ since $\dB$ is benign.
   When the simulation of $\A$ reaches the point 
   just after the announcement of $f_0$ and $f_1$ in Step~3, 
   $\dhB$ finds $k'$ such that $d_H(\hat{\theta}|_{I_{k'}},\theta|_{I_{k'}})$ 
   is minimum
   for $k'\in\{0,1\}$. $\dhB$ then calls ${\fot}$ with input $k'$ and
   obtains output $s_{k'}$. $\dhB$ sets  
   $m'_{k'} = s_{k'} \oplus f_{k'}(x|_{I_{k'}})$ 
   and $m'_{1-k'} \in_R\{0,1\}^\ell$ before sending $(m_0,m_1)$
   to $\dB$. $\dhB$ then outputs whatever $\dB$ outputs.
 
 We now argue that the state output by $\dhB$ is statistically close to
 the state output by $\dB$ when executing $\QOT$ with the real $\A$.
 The only difference is that while $\dhB$ outputs
 $m'_{1-k'}\in_R\{0,1\}^\ell$, $\dB$ outputs $m_{1-k'} = s_{1-k'}
 \oplus f_{1-k'}(x|_{I_{1-k'}})$.  To conclude, we simply need to show
 that $m_{1-k'}$ is statistically indistinguishable from uniform from
 the point of view of $\dB$.  Note that since $\theta$ and
 $\hat{\theta}$ are independent and $\theta$ is a uniform $n$-bit
 string, we have that for any $\epsilon>0$, $d_H(\theta,\hat{\theta})>
 (1-\epsilon)n/2$, except with negligible probability. It
 follows that with overwhelming probability
 $d_H(\theta|_{I_{1-k'}},\hat{\theta}|_{I_{1-k'}})\geq (1-\epsilon)n/4$. Since $\dB$ is $\beta$-benign, 
 we have that $\Hoo\bigl(X|_{I_{1-k'}} \,\big|\, X|_{I_{k'}} = x|_{I_{k'}}\bigr) \geq 
 (1-\epsilon)n/4-\beta n$ and $H_0(\rho_E)\leq \beta n$
 which implies, from
 privacy amplification, 
 that $f_{1-k'}(x|_{I_{1-k'}})$ is statistically indistinguishable
 from uniform for $\dB$ provided $\frac{\ell}{n} < \frac{1}{4}-2\beta-\epsilon$
 for any $\epsilon>0$. We conclude that 
 $m_{1-k'}$ is statistically close to uniform. \qed 
 \end{proof}

By combining Theorem~\ref{thm:OT} with Theorem~\ref{thm:Compiler}, and the results of~\cite{DFRSS07} (realizing that the same analysis also applies to $\QOT$) with Theorem~\ref{thm:BQSM}, we obtain the following hybrid-security result. 

 \begin{corollary}
 Let $0<\alpha<1$ and $\lambda < \frac18$. Then protocol $\compile(\QOT)$ is computationally secure against dishonest Bob and unconditionally secure against $\gamma (1\!-\!\alpha)$-BQSM Bob with $\gamma < \frac14 - 2 \lambda$. 
 \end{corollary}

\subsection{Password-Based Identification}
\label{sec:QID}

We want to apply our compiler to the quantum password-based identification scheme from~\cite{DFSS07}. Such an identification scheme allows a user $\A$ to identify herself to server $\B$ by means of a common (possibly non-uniform and low-entropy) password $w \in \cal W$, such that dishonest $\dA$ cannot delude honest server $\B$ with probability better then trying to guess the password, and dishonest $\dB$ learns no information on $\A$'s password beyond trying to guessing it and learn whether the guess is correct or not. 

In~\cite{DFSS07}, using quantum-information-theoretic security definitions, the proposed identification scheme was proven to be unconditionally secure against arbitrary dishonest Alice and against quantum-memory-bounded dishonest Bob. In~\cite{FS09} it was then shown that these security definitions imply simulation-based security as considered here, with respect to the functionality $\F_{\ID}$ given in Figure~\ref{fig:identfunc}.%
\footnote{Actually, the definition and proof from~\cite{DFSS07} guarantees security only for a slightly weaker functionality, which gives some unfair advantage to dishonest $\dA$ in case she guesses the password correctly; however, as discussed in~\cite{FS09}, the protocol from~\cite{DFSS07} does implement functionality $\F_{\ID}$. }

\begin{figure}
\begin{framed}
{\bf Functionality} $\F_{\ID}$: \; 
Upon receiving $w_A,w_B \in\cal W$ from user Alice and from server
Bob, respectively, $\F_\ID$ outputs the bit \smash{$y \assign (w_A \eqq w_B)$} to Bob. 
In case Alice is dishonest, she may choose $w_A = \,\perp$ (where
$\perp \, \not\in \cal W$). For any choice of $w_A$ the bit $y$ is
also output to dishonest Alice.
\vspace{-1ex}
\end{framed}
\vspace{-2ex}
\caption{The Ideal Password-Based Identification Functionality.}\label{fig:identfunc}
\vspace{-1ex}
\end{figure}

We cannot directly apply our compiler to the identification scheme as given in~\cite{DFSS07}, since it is {\em not} a BB84-type protocol. The protocol does start with a preparation phase in which Alice sends BB84 qubits to Bob, but Bob does not measure them in a random basis but in a basis determined by his password $w_B \in {\cal W}$; specifically, Bob uses as basis the encoding $\mathfrak{c}(w_B)$ of $w_B$ with respect to a code $\mathfrak{c}:{\cal W} \rightarrow \set{+,\x}^n$ with ``large" minimal distance. 
However, it is easy to transform the original protocol from~\cite{DFSS07} into a BB84-type protocol without affecting security: We simply let Bob apply a {\em random shift} $\kappa$ to the code, which Bob only announces to Alice in the post-processing phase, and then Alice and Bob complete the protocol with the shifted code. The resulting protocol $\QID$ is described in Figure~\ref{fig:ident}, where $\cal F$ and $\cal G$ are suitable families of (strongly) universal hash functions (we refer to~\cite{DFSS07} for the exact specifications). It is not hard to see that this modification does not affect security as proven in~\cite{DFSS07} (and~\cite{FS09}). 

\begin{figure}
\begin{framed}
\noindent\hspace{-1.5ex} {\sc Protocol $\QID:$ } \\[-4ex]
\begin{description}\setlength{\parskip}{0.5ex}
\item[{\it Preparation:}] $\A$ chooses $x \in_R \set{0,1}^n$ and $\theta \in_R \set{+,\x}^n$ and sends $\ket{x}_{\theta}$ to~$\B$, and $\B$ chooses $\hat{\theta} \in_R \set{0,1}^n$ and obtains $\hat{x} \in \set{0,1}^n$ by measuring $\ket{x}_{\theta}$ in basis $\hat{\theta}$. 
\item[{\it Post-processing:}] 
\begin{enumerate}
\item
$\B$ computes a string $\kappa\in \{ +, \times \}^n$ such that
$\hat{\theta}= \mathfrak{c}(w) \oplus \kappa$ (we think of $+$ as 0
and $\times$ as 1 so that $\oplus$ makes sense). He sends $\kappa$ to
$\A$ and we define $\mathfrak{c'}(w) \assign \mathfrak{c}(w) \oplus \kappa$. 
\item 
$\A$ sends $\theta$ and $f \in_R \mathcal{F}$ to $\B$. 
Both compute $I_w \assign \{ i: \theta_i=\mathfrak{c'}(w)_i \}$.
\item 
$\B$ sends $g \in_R \mathcal{G}$.
\item 
$\A$ sends $z\assign f(x|_{I_w}) \oplus g(w)$ to $\B$. 
\item $\B$ accepts if and only if $z=f(\hat{x}|_{I_w}) \oplus g(w)$.
\end{enumerate}
\end{description}
\vspace{-1.5ex}
\end{framed}
\vspace{-1.5ex}
 \caption{Protocol for Secure Password-Based Identification}
 \label{fig:ident} 
\end{figure}

\begin{theorem}\label{thm:ID}
If the code $\mathfrak{c}:{\cal W} \rightarrow \set{+,\x}^n$ can correct at least $ \delta n$ errors in polynomial-time for a constant $\delta$, then protocol $\QID$ is unconditionally secure against $\beta$-benign Bob for any $\beta < \frac{\delta}{4}$. 
\end{theorem}

\begin{proof}
For any given benign Bob $\dB$, we construct $\dhB$ as follows. $\dhB$ runs locally a copy of $\dB$ and simulates Alice's actions by running $\A$ faithfully except for the following modifications. After the preparation phase, $\dhB$ gets $\hat{\theta}$ and $\kappa$ from $\dB$ and attempts to decode
$\hat{\theta} \oplus \kappa$. If this succeeds, it computes $w'$ such that $\mathfrak{c}(w')$ is the decoded codeword. Otherwise an arbitrary $w'$ is chosen.
Then, $\dhB$ submits $w'$ as Bob's input $w_B$ to $\F_{\ID}$ and receives output $y \in \set{0,1}$. If $y = 1$ then $\dhB$ faithfully completes $\A$'s simulation using $w'$ as $w$; else, $\dhB$ completes the simulation by using a random $z'$ instead of $z$. In the end, $\dhB$ outputs whatever $\dB$ outputs. 

We need to show that the state output by $\dhB$ (respectively $\dB$) above is statistically close to the state output by $\dB$ when executing $\QID$ with real $\A$. 
Note that if $w' = w_A$, then the simulation of $\A$ is perfect and thus the two states are equal. If $w' \neq w_A$ then the simulation is not perfect: the real $\A$ would use $z= f(x|_{I_{w_A}}) \oplus g(w_A)$ instead of random $z'$. It thus suffices to argue that $f(x|_{I_w})$ is statistically close to random and independent of the view of $\dB$ for any fixed $w \neq w'$. 
Note that this is also what had to be proven in~\cite{DFSS07}, but under a different assumption, namely that $\dB$ has bounded quantum memory, rather than that he is benign; nevertheless, we can recycle part of the proof. 

Recall from the definition of a benign Bob that the common state after the preparation phase is statistically close to a state for which it is guaranteed that $\Hoo(X|_I) \geq d_H(\theta|_I,\hat{\theta}|_I) - \beta n$ for any $I \subseteq \set{1,\ldots,n}$, and $\Hmax(\rho_{ER}) \leq \beta n$. 
By the closeness of these two states, switching from the real state to the ``ideal'' state (which satisfies these bounds) has only a negligible effect on the state output by $\dhB$; thus, we may assume these bounds to hold. 

Now, if decoding of $\hat{\theta}\oplus\kappa$ succeeded, it is at Hamming distance at most
$\delta n$ from $\mathfrak{c}(w')$. Since the distance from here to the (distinct) codeword
$\mathfrak{c}(w)$ is greater than $2\delta n$, we see that  $\hat{\theta}\oplus\kappa$ is at least
$\delta n$ away from $\mathfrak{c}(w)$. The same is true if decoding failed, since then
 $\hat{\theta}\oplus\kappa$ is at least $\delta n$ away from any codeword.
It follows that $\mathfrak{c}'(w) = \mathfrak{c}(w) \oplus \kappa$ has Hamming distance at least $\delta n$ 
from $\hat{\theta}$. Furthermore, for arbitrary $\eps > 0$ and except with negligible probability, the Hamming distance between $\theta|_{I_w} = \mathfrak{c}'(w)|_{I_w}$ and $\hat{\theta}|_{I_w}$ is at least essentially $(\delta/2 - \eps) n$. Therefore, we can conclude that $\Hoo(X|_{I_w}) \geq (\delta/2 - \eps  - \beta) n$ and $\Hmax(\rho_{ER}) \leq \beta n$. 
But now, if such bounds hold such that $\Hoo(X|_{I_w}) - \Hmax(\rho_{ER})$ is positive and linear in $n$, which is the case here by the choice of parameters, then we can step into the proof from~\cite{DFSS07} and conclude by privacy amplification~\cite{RK05} that $z$ is close to random and independent of $E$. 
This finishes the proof. 
\qed
\end{proof}
By combining Theorem~\ref{thm:ID} with Theorem~\ref{thm:Compiler}, and the results of~\cite{DFSS07} with Theorem~\ref{thm:BQSM}, we obtain the following hybrid-security result. 

\begin{corollary}
  Let $0<\alpha<1$ and $|{\cal W}| \leq 2^{\nu n}$. If the code
  $\mathfrak{c}:{\cal W} \rightarrow \set{+,\x}^n$ can correct
  $ \delta n$ errors for a constant  $\delta > 0$ in polynomial-time,
  then protocol $\compile(\QID)$ is computationally secure
  against dishonest Bob and unconditionally secure against $\gamma
  (1\!-\!\alpha)$-BQSM Bob with $\gamma < \frac{\delta}{2} - \nu$.
\end{corollary}
Families of codes as required in these results, correcting a constant fraction of errors efficiently and with constant information rate are indeed known, see \cite{SS96}.

In the next section, we briefly discuss
how to obtain
hybrid security against {\em man-in-the-middle} attacks by means of
incorporating the techniques used in~\cite{DFSS07} to obtain security
in the BQSM against such attacks.

 \subsection{Protecting against Man-in-the-middle Attacks}\label{sec:MIM}

 The compiled quantum protocols from Sections~\ref{sec:QOT} and \ref{sec:QID} protect
 against (arbitrary) dishonest Alice and against (computationally or
 quantum-storage bounded) dishonest Bob. However, in particular in the
 context of identification, it is also important to protect against a
 {\em man-in-the-middle} attacker, Eve, who attacks an execution of the
 protocol with honest parties $\A$ and $\B$ while having full control
 over the classical and the quantum communication. Both, $\QID$ and
 $\compile(\QID)$, are insecure in this model: Eve might measure one of
 the transmitted qubits, say, in the $+$ -basis, and this way learn
 information on the basis $\hat{\theta}_i$ used by $\B$ and thus on the
 password $w$ simply by observing if $\B$ accepts or rejects in the
 end.

 In~\cite{DFSS07} it was shown how to enhance $\QID$ in order to obtain
 security (in the bounded-quantum-storage model) against
 man-in-the-middle attacks. The very same techniques can also be used
 to obtain {\em hybrid security} against man-in-the-middle attacks for
 $\compile(\QID)$. The techniques from~\cite{DFSS07} consist of the
 following two add-on's to the original protocol. (1) Checking of a
 random subset of the qubits in order to detect disturbance of the
 quantum communication; note that $\compile(\QID)$ already does such a
 check, so this is already taken care of here. And (2) authentication of
 the classical communication. This requires that Alice and Bob, in
 addition to the password, share a high-entropy key $k$ that could be
 stored, e.g., on a smart-card.  This key will be used for a so-called
 {\em extractor MAC} which has the additional property, besides being a
 MAC, that it also acts as an extractor, meaning if the message to be
 authenticated has high enough min-entropy, then the key-tag pair is
 close to randomly and independently distributed. As a consequence, the
 tag gives away (nearly) no information on $k$ and thus $k$ can be
 re-used in the next execution of the protocol.%
 \footnote{This is in
   %sharp 
   contrast to the standard way of authenticating the classical
   communication, where the authentication key can only be used a
   bounded number of times. }

 Concretely, in order to obtain hybrid-security against
 man-in-the-middle attacks for $\compile(\QID)$, $\A$ will, in her last
 move of the protocol, use an extractor MAC to compute and send to $\B$
 an authentication tag, computed on all the classical messages
 exchanged plus the string $x|_{I_w}$. This tag, together with the
 qubit checks, prevents Eve from interfering with the (classical and
 quantum) communication without being detected, and security against
 Eve essentially follows from the security against impersonation
 attacks.  Note that including the $x|_{I_w}$ into the authenticated
 message guarantees the necessary min-entropy, and as such the
 re-usability of the key $k$.

We emphasize that the protocol is still secure against impersonation attacks (i.e. dishonest Alice or Bob) even if the adversary knows $k$. We omit formal proofs since they literally follow the corresponding proofs in~\cite{DFSS07}.

\section*{Acknowledgments}
We thank Dominique Unruh for useful comments about the efficiency of
the simulators.

SF is supported by the Dutch Organization for Scientific Research (NWO). CL acknowledges financial support by the MOBISEQ research project funded by NABIIT, Denmark. LS is supported by the QUSEP project of the Danish Natural Science Research Council, and by the QuantumWorks Network. CS acknowledges support by EU fifth framework project QAP IST 015848 and a NWO VICI grant 2004-2009.

\addcontentsline{toc}{chapter}{Bibliography} 
\bibliographystyle{alpha}
\bibliography{crypto,qip,procs}

\appendix

\AppendixOnOff{

\section{Sequential Composition Theorem for Computational Security}
\label{app:composition}

In this appendix, we show that our new
Definition~\ref{def:polyboundedBobCRS} of computational security
allows for sequential composability in a classical environment. In
order to state the composition theorem, we need to define what we mean
by running a quantum protocol in a classical environment. Again, we
give here a brief summary of the setting from~\cite{FS09} and refer
the interested reader to the original article for further details.

A classical two-party {\em oracle protocol}%
\footnote{In~\cite{FS09}, the
  more standard term \emph{hybrid protocol} is used, but as this term
  is used differently in this paper, we avoid it here in the context
  of composability.} $\Sigma^{\F_1\cdots\F_\ell} =
(\hA,\hB)$ between Alice and Bob is a protocol which makes a bounded
number $k$ of sequential oracle calls to possibly different ideal
functionalities $\F_1,\ldots,\F_\ell$.

For the oracle protocol to be {\em classical}, we mean that it has
classical in- and output (for the honest players), but also that all
communication between Alice and Bob is classical.%
\footnote{We do not
  explicitly require the internal computations of the honest parties
  to be classical. } 
Consider a dishonest player, say Bob, and consider the common state
\smash{$\rho_{U_j V'_j}$} at any point during the execution of the
oracle protocol when a call to functionality $\F_i$ is made. The
requirement for the oracle protocol to be \emph{classical} is now
expressed in that there exists a classical $Z_j$---to be understood as
consisting of $\dhB$'s classical communication with $\hA$ and with the
$\F_{i'}$'s up to this point---such that given $Z_j$, Bob's quantum
state $V'_j$ is not entangled with Alice's classical input and
auxiliary information: \smash{$\rho_{U_j Z_j V'_j} = \rho_{\MC{U_j}{Z_j}{V'_j}}$}.  Furthermore, we require that we may assume
$Z_j$ to be part of $V'_j$ in the sense that for any $\dhB$ there
exists $\dhB'$ such that $Z_j$ is part of $V'_j$. This definition is
motivated by the observation that if Bob can communicate only
classically with Alice, then he can entangle his quantum state with
information on Alice's side only by means of the classical
communication.

We also consider the protocol we obtain by replacing the ideal
functionalities by quantum two-party sub-protocols
$\pi_1,\ldots,\pi_\ell$ with classical in- and outputs for the honest
parties: whenever $\Sigma^{\F_1\cdots\F_\ell}$ instructs $\hA$ and
$\hB$ to execute ${\F_i}_{\hA,\hB}$, they instead execute $\qp_i$ and
take the resulting outputs. We write $\cp^{\qp_1\cdots\qp_\ell} =
(\A,\B)$ for the real quantum protocol we obtain this way. 

We recall that in order to define computational security against a
dishonest Bob in the common-reference-string model, we considered a
polynomial-size quantum circuit, called \emph{input sampler}, which
takes as input the security parameter $m$ and the CRS $\crs$ (chosen
according to its distribution) and which produces the input state
$\rho_{U Z V'}$; $U$ is Alice's classical input to the protocol, and
$Z$ and $V'$ denote the respective classical and quantum information
given to dishonest Bob. We require from the input sampler that
$\rho_{U ZV'} = \rho_{\MC{U}{Z}{V'}}$, i.e., that $V'$ is correlated
with Alice's part only via the classical~$Z$. When considering
classical hybrid protocols $\cp^{\qp_1\cdots\qp_\ell}$ in the real
world, where the oracle calls are replaced with quantum protocols
using a common reference string, it is important that every real
protocol $\pi_i$ uses a separate instance (or part) of the common
reference string which we denote by $\crs_i$.

\begin{theorem}[Sequential Composition] \label{thm:Comp}
  Let $\Sigma^{\F_1\cdots\F_\ell} = (\hA,\hB)$ be a classical
  two-party oracle protocol which makes at most $k=\poly(n)$ oracle
  calls to the functionalities, and for every $i \in
  \set{1,\ldots,\ell}$, let protocol $\pi_i$ be a computationally
  secure implementation of $\F_i$ against $\dBobPoly$. 

  Then, for every real-world adversary $\dB \in \dBobPoly$ who
  accesses the common reference string $\crs=\crs_1,\ldots,
  \crs_k$ there exists an ideal-world adversary $\dhB \in \dBobPoly$
  who does not use $\crs$ such that for every efficient input sampler, it
  holds that 
  the outputs in the real and ideal world are q-indistinguishable,
  i.e.~
$$
out_{\A,\dB}^{\Sigma^{\pi_1\cdots\pi_\ell}} \approxq out_{\hA,\dhB}^{\Sigma^{\F_1\cdots\F_\ell} }
$$
\end{theorem}

Note that we do not specify what it means for the oracle protocol
to be secure; Theorem~\ref{thm:Comp} guarantees that {\em whatever}
the oracle protocol achieves, an indistinguishable output is produced
by the real-life protocol with the oracle calls replaced by protocols.
But of course in particular, if the oracle protocol {\em is} secure in
the sense of Definition~\ref{def:polyboundedBobCRS}, then so is the
real-life protocol:

\begin{corollary}
  If $\Sigma^{\F_1\cdots\F_\ell}$ is a computationally secure
  implementation of $\cal G$ against $\dBobPoly$, and if $\pi_i$ is a
  computationally secure implementation of $\F_i$ against $\dBobPoly$
  for every $i \in \set{1,\ldots,\ell}$, then
  $\Sigma^{\pi_1\cdots\pi_\ell}$ with at most $k=\poly(n)$ oracle
  calls is a computationally secure implementation of~$\cal G$ against
  $\dBobPoly$.
\end{corollary}

\def\mark#1{\bar{#1}}

 The following proof is an adaptation of the
 sequential-composability proof in the information-theoretical setting
 given in~\cite{FS09}.

\begin{proof}[of Theorem~\ref{thm:Comp}]
Consider a dishonest $\dB \in \dBobPoly$. 
We prove the claim by induction on $k$. If no oracle calls are made,
we can set $\dhB \assign \dB$ and the claim holds trivially. Consider
now a protocol $\Sigma^{\F_1\cdots\F_\ell}$ with at most $k > 0$
oracle calls.  For simplicity, we assume that the number of oracle
calls equals $k$, otherwise we instruct the players to makes some
``dummy calls''.  Let $\rho_{U_k Z_k V'_k}$ be the common state right
before the $k$-th and thus last call to one of the sub-protocols
$\pi_1,\ldots,\pi_\ell$ in the execution of the real protocol
$\Sigma^{\pi_1,\ldots,\pi_\ell}$. To simplify notation in the rest of
the proof, we omit the index $k$ and write \smash{$\rho_{\mark{U}
    \mark{Z} \mark{V}'}$} instead; see Figure~\ref{fig:compproof}. We
know from the induction hypothesis for $k-1$ that there exists an
ideal-world adversary $\dhB \in \dBobPoly$ not using the common
reference string such that $\rho_{\mark{U} \mark{Z} \mark{V}'}
\approxq \sigma_{\mark{U} \mark{Z} \mark{V}'}$ where $\sigma_{\mark{U} \mark{Z} \mark{V}'}$ is the common state right before the
$k$-th call to a functionality in the execution of the oracle protocol
$\Sigma^{\F_1\cdots\F_\ell}_{\hA,\dhB} \rho_{U Z V'}$. As described
at the begin of this section, $\mark{U}$ and
$\mark{Z},\mark{V}'$ are to be understood as follows.
$\mark{U}$ denotes $\A$'s (respectively $\hA$'s) input
to the sub-protocol (respectively functionality) that is to be called
next. $\mark{Z}$ collects the classical communication dictated by
$\Sigma^{\F_1\ldots,\F_\ell}$ as well as $\dhB$'s classical inputs to
and outputs from the previous oracle calls and $\mark{V}'$ denotes the
dishonest player's current quantum state. Note that the existence of
$\mark{Z}$ is guaranteed by our formalization of {\em classical}
oracle protocols and $\sigma_{\mark{U} \mark{Z} \mark{V}'} =
\sigma_{\MC{\mark{U}}{\mark{Z}}{\mark{V}'}}$.

Let $\crs_i$ be the common reference string used in protocol $\pi_i$.
For simplicity, we assume
that the index $i$, which determines the sub-protocol $\pi_i$
(functionality~$\F_i$) to be called next, is {\em fixed}
and we just write $\pi$ and $\F$ for $\pi_i$ and $\F_i$,
respectively. 

\begin{figure}
 \centering 
 { \input{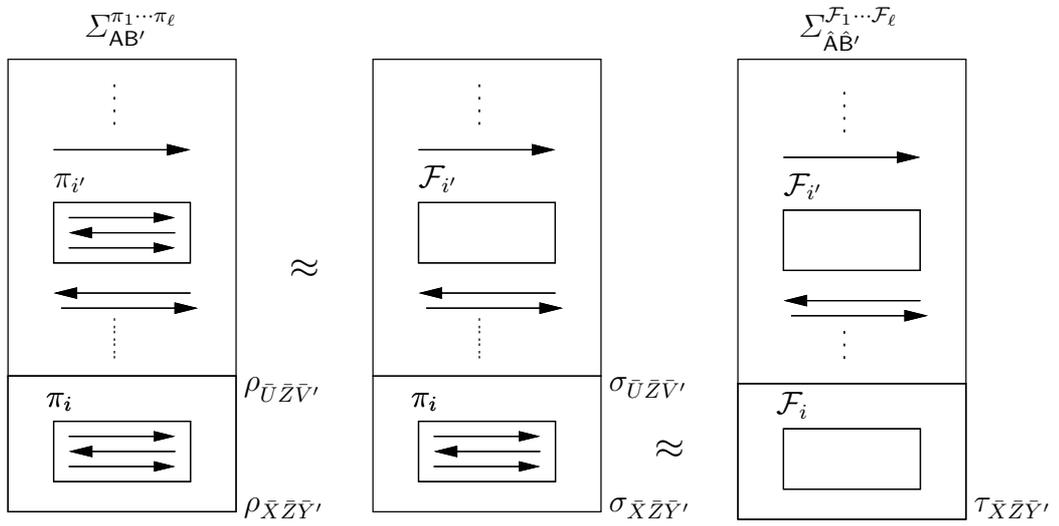} }
 \caption{Steps of the Composability Proof} 
 \label{fig:compproof} 
\end{figure}

It follows from Definition~\ref{def:polyboundedBobCRS} of computational
security that there exists $\dhB \in \dBobPoly$ (independent of the
input state) not using $\crs_i$ such that the corresponding output states
$\sigma_{\mark{X}\mark{Z}\mark{Y}'}$ and
$\tau_{\mark{X}\mark{Z}\mark{Y}'}$ produced by $\F_{\hA,\dhB}$
(as prescribed by the oracle protocol) and $\pi_{\A,\dB}$ run on the
state $\sigma_{\mark{U} \mark{Z} \mark{V}'} =
\sigma_{\MC{\mark{U}}{\mark{Z}}{\mark{V}'}}$ are
q-indistinguishable.

The induction step is then completed as follows.
\[ out_{\A,\dB}^{\Sigma^{\pi}} = \rho_{\mark{X}\mark{Z} \mark{Y}'} = (\pi_{\A,\dB})\,
\rho_{\mark{U}\mark{Z}\mark{V}'} \approxq
(\pi_{\A,\dB})\,\sigma_{\mark{U}\mark{Z}\mark{V}'} =
\sigma_{\mark{X}\mark{Z}\mark{Y}'} \approxq
\tau_{\mark{X}\mark{Z}\mark{Y}'} = out_{\hA,\dhB}^{\Sigma^{\F}}
\]

Note that the strategy of $\dhB$ does not depend on the state
$\sigma_{\mark{U}\mark{Z}\mark{V}'}$, and hence, the overall
ideal-world adversary $\dhB$ does not depend on the input state
either. Furthermore, the concatenation of two polynomially bounded
players is polynomially bounded, i.e.~$\dhB \in \dBobPoly$.

\qed
\end{proof}

}

\end{document}